\documentclass{amsart}
\usepackage{amsmath}
\usepackage{amsfonts}
\usepackage{geometry}

\newtheorem{theorem}{Theorem}
\theoremstyle{plain}

\newtheorem{conjecture}{Conjecture}

\newtheorem{problem}{Problem}

\numberwithin{equation}{section}
\input{tcilatex}

\begin{document}
\title[Permanent and Entanglement]{Matrix permanent and quantum entanglement
\\
of permutation invariant states}
\author{Tzu-Chieh Wei}
\thanks{Institute for Quantum Computing and Department of Physics and
Astronomy, University of Waterloo, Waterloo N2L 3G1, Canada and Department of
Physics and Astronomy, University of British Columbia, Vancouver, BC V6T 1Z1,
Canada; twei@phas.ubc.cai}
\author{Simone Severini}
\thanks{Institute for Quantum Computing \ and Department of Combinatorics \&
Optimization, University of Waterloo, Waterloo N2L 3G1, Canada, 
and Department of Physics and Astronomy, University College London, WC1E 6BT London, United Kingdom;
simoseve@gmail.com\\
\emph{Acknowledgements.} Thanks to David Gross, Otfried G\"{u}hne, and
Shashank Virmani, for fruitful discussion. This work was supported by
DTOARO, ORDCF, CFI, CIFAR, and MITACS. }
\maketitle

\begin{abstract}
We point out that a geometric measure of quantum entanglement is related to
the matrix permanent when restricted to permutation invariant states. This
connection allows us to interpret the permanent as an angle between vectors.
By employing a recently introduced permanent inequality by Carlen, Loss and
Lieb, we can prove explicit formulas of the geometric measure for permutation
invariant basis states in a simple way.
\end{abstract}

\section{Introduction}

In the Editor's statement forewording the 1982 monograph \emph{Permanents} by
Minc~\cite{m82}, Gian-Carlo Rota wrote the following words:

\begin{quotation}
\textquotedblleft A permanent is an improbable construction to which we
might have given little chance of survival fifty years ago. Yet numerous
appearances it has made in physics and in probability betoken the mystifying
usefulness of the concept, which has a way of recurring in the most
disparate circumstances.\textquotedblright
\end{quotation}

The present paper highlights a connection between the permanent and
entanglement of certain quantum states, therefore putting in evidence a
further appearance of the permanent in physics. Some background is useful
for delineating the context. The \emph{permanent} \ and the \emph{%
determinant }of an $n\times n$ matrix $A$ (with entries in a commutative
ring) are respectively defined as%
\begin{equation*}
\begin{tabular}{lll}
$\text{perm}(A)=\dsum\limits_{\pi \in S_{n}}\dprod\limits_{i=1}^{n}A_{i,\pi
(i)}$ & and & $\det \left( A\right) =\dsum\limits_{\pi \in S_{n}}\left(
-1\right) ^{\text{sgn}\left( \pi \right) }\dprod\limits_{i=1}^{n}A_{i,\pi
(i)},$%
\end{tabular}%
\end{equation*}%
where $S_{n}$ denotes the full symmetric group on a set of $n$ symbols. In
light of the definitions, it is reasonable that there are cases, as it was
originarily observed in 1859 by Cayley~\cite{c59}, in which permanents can
be computed by means of determinants (see also, \emph{e.g.}, Kasteleyn~\cite%
{k61}, Godsil and Gutman~\cite{gg81}, and Frieze and Jerrum~\cite{fj95}).
Geometrically, the determinant is the volume of the parallelepiped defined by
the lines of the matrix; algebraically, it is the product of all eigenvalues
including their multiplicities. It is curious that despite its striking
similarity with the determinant, the permanent does not have any known
geometric or algebraic interpretation. Moreover, while standard Gaussian
elimination provides an efficient technique for computing the determinant, the
exact computation of the permanent remains a notoriously difficult problem.

The best known algorithm for an $n\times n$ matrix, due to Ryser in
1963~\cite{r63}, needs $\Theta (n2^{n})$ operations. By a seminal result of
Valiant~\cite{val79}, computing the permanent is indeed \textquotedblleft The\
$\#P$-hard problem\textquotedblright . Thus computing the permanent on worst
case inputs cannot be done in polynomial time unless $P^{\#P}=P$ and in
particular $P=NP$. It follows that algorithms for the permanent acquired a
special position in computational complexity to the point of becoming a
fertile ground for many approaches, including iterative balancing (Linial
\emph{et al.}~\cite{l00}), elementary recursive algorithms (Rasmussen~\cite%
{r94}) and, most relevantly, Markov chain Monte Carlo methods (Broder~\cite%
{b86}, Jerrum and Sinclair~\cite{js89}). These efforts contributed to a deeper
understanding of the permanent and produced important results as fully
polynomial randomized approximation schemes for several types of matrices (see
Jerrum, Sinclair and Vigoda~\cite{jsv01}, Barvinok~\cite{b99}, and the
references contained therein).

From the mathematical perspective, there are two main lines of research
centered on the permanent: the study of permanents and probability (see,
\emph{e.g.}, Friedland \emph{et al.}~\cite{f04}); min/max questions concerning
a number of inequalities, peaking with the proofs by Egorychev~\cite{e81} and
Falikman~\cite{f81} (see also~\cite{m82a}) of the 1926 van der Waerden's
conjecture about permanents of doubly stochastic matrices.

When focusing our attention on physics, we may divide into at least four
groups the known applications of the permanent: the many aspects of the
dimer problem, some uses involving the hafnian, Monte Carlo generators, and,
more pertinent to our discussion, linear optical networks for quantum
information processing. Intuitively, the permanent tends to be related to
bosons while the determinant to fermions. A list of applications follows.

The problem of computing the permanent of a $(0,1)$-matrix is the same as
the problem of counting the number of perfect matchings in a bipartite
graph. This translates into the dimer problem \cite{k61, hl72}. The problem
is traditionally related to models for adsorption of diatomic molecules on
crystal surfaces, mixtures of molecules of different sizes and the
cell-cluster theory of the liquid state (see, \emph{e.g.}, Welsh~\cite{w90}%
). New applications occur in the study of configurations of melting crystals
(Okounkov \emph{et al.}~\cite{o06}), BPS black holes (Heckman and Vafa~\cite%
{hv07}), and quiver gauge theories (Hanany and Kennaway~\cite{h05}).

The \emph{hafnian} of a matrix $A$, denoted by hf$(A)$, is a polynomial
generalizing the permanent as the pfaffian generalizes the determinant. This
notion was originarily introduced by Caianiello to express the perturbation
expansions of boson field theories~\cite{c53}.

Bia\l as and Krzywicki~\cite{bk95} (see also Wosiek~\cite{w97}) introduced a
procedure to include Bose-Einstein correlations in Monte Carlo event
generators. The procedure makes use of the generalized Wigner functions and it
requires to compute the permanent of a correlation matrix depending on
particle momenta and other parameters.

Sheel \emph{et al.}~\cite{s04} have proved that matrix elements of unitarily
transformed photonic multi-mode states can be written as permanents associated
with the symmetric tensor power of the beam splitter matrix (see also Kok
\emph{et al.}~\cite{k07}). The result implies that computing matrix elements
in the Fock basis is not an easy task. Permanents count the ways of
redistributing $n$ single photons through an SU$\left( n\right) $ network to
yield exactly $n$ single photons at the outputs and then allow to compute
probability amplitudes~\cite{s03, s06}; optimal networks are obtained by
maximizing the permanent under given constraints.

The present paper is concerned with quantum entanglement. It is now
established that entanglement is an important physical quantity whose presence
appears to be necessary in many applications of quantum information
processing~\cite{nc}. Given a generic quantum state, detecting and measuring
its entanglement is a challenge from both the mathematical and the
experimental point of view. It is computationally hard~\cite{gu} and
furthermore there is no general agreement on how to quantify entanglement. A
number of different entanglement measures have been therefore introduced using
a variety of approaches (see the Horodeckis' review~\cite{horo} and its
references).

Here we consider a \emph{geometric measure}. This quantifies the angle between
two states, subject to an optimization problem: the closest state with zero
entanglement and the state under analysis. It turns out that, when we restrict
the analysis to permutation invariant states, such a measure can be described
in terms of the permanent of a certain matrix. A recent permanent inequality
of Carlen, Loss and Lieb~\cite{clb06} (see also Samorodnitsky~\cite{sa08}) can
be employed to bypass the optimization problem. Hence, maximizing the
permanent of a matrix with certain constraints, reduces to a simpler problem.
As we have mentioned earlier, the permanent does not have any known geometric
interpretation. The present work contributes in this direction, by
interpreting the permanent as an angle between vectors (or a cosine, to be
more precise).

We remark on a historic note that was pointed out to us by an anonymous
referee regarding Conjecture 1 (see below).  A special case of the proof was
considered in Ref.~\cite{ni}. This manuscript~\cite{we09} contains a proof for
situations in which the symmetric state has an expansion (in either the
computational basis or the symmetric basis $|S(n,\vec{k})\rangle$ defined
below) which requires only positive coefficients. This was shown independently
in Ref.~~\cite{hai} by a completely different method. Subsequently, a more
general proof (with extended results) using elementary matrix theory appeared
in Ref.~\cite{hu}, in which it was recognized that a proof was already given
earlier in the context of the theory of tensor norms and homogeneous
polynomials over Banach spaces, e.g., see Ref.~\cite{di}.

The rest of the paper is structured as follows. In Section 2, we define the
geometric measure of entanglement. In Section 3, we state and prove our
results. For the sake of clarity, we shall also include a few examples.

\section{Geometric measure of entanglement}

The geometric measure of entanglement used here was firstly introduced by
Shimony~\cite{s95} in the setting of bipartite pure states. It was generalized
to the multipartite states by Barnum and Linden~\cite{bl01}, and further
extended by Wei and Goldbart~\cite{wg03}. The intuition beyond the measure
consists of thinking about entanglement as an angle between two states:
namely, the state in analysis and a product state, \emph{i.e.}, a state with
zero entanglement. Crucially, the product state is chosen over all possible
product states so that it minimizes the angle with the state in analysis.

Let
\begin{equation}
|\psi \rangle =\sum_{p_{1}\cdots p_{n}}\chi _{p_{1}p_{2}\cdots
p_{n}}|e_{p_{1}}^{(1)}e_{p_{2}}^{(2)}\cdots e_{p_{n}}^{(n)}\rangle
\label{pur}
\end{equation}%
be a generic multipartite pure state in a Hilbert space $\mathcal{H}\cong
\bigotimes_{k_{i}}\mathbb{C}_{i}^{k_{i}}$ of dimensionality $\dim \left(
\mathcal{H}\right) =\prod_{i=1}^{n}k_{i}$. Each set $\{|e_{p_{i}}^{(i)}%
\rangle :p_{i}=1,2,...,k_{i}\}$ is a local basis for the $i$-th subspace $%
\mathbb{C}_{i}^{k_{i}}$. A pure state in $\mathcal{H}$ is said to be \emph{a
product state} if it can be written in the form%
\begin{equation*}
|\phi \rangle =\dbigotimes\limits_{i=1}^{n}|\phi ^{(i)}\rangle \equiv
\dbigotimes\limits_{i=1}^{n}\sum_{p_{i}}\left(
c_{p_{i}}^{(i)}|e_{p_{i}}^{(i)}\rangle \right) ,
\end{equation*}%
where $|\phi ^{(i)}\rangle \in \mathbb{C}_{i}^{k_{i}}$ is some pure state; $%
|\psi \rangle $ is said to be \emph{entangled}, otherwise. Given a state $%
|\psi \rangle $ as in Eq.~(\ref{pur}), let us define%
\begin{equation*}
\Lambda _{\max }(\psi ):=\max_{|\phi \rangle }|\langle \phi |\psi \rangle |,
\end{equation*}%
where the maximization is performed over all product states $|\phi \rangle
\in \mathcal{H}$. This formula tells us how well the possibly entangled
state $|\psi \rangle $ can be approximated by a product state. The formula
provides a method that satisfies the \emph{desiderata} for a well-defined
measure of entanglement~\cite{horo}. Such a method is usually called \emph{%
geometric measure}. The terminology is justified since $\Lambda _{\max }({%
\psi })$ is an angle between two vectors. Hence, notice that the amount of
entanglement increases while $\Lambda _{\max }({\psi })$ decreases and
therefore the quantity of entanglement depends essentially on $\Lambda
_{\max }(\psi )$. Given a state $|\psi \rangle $, concrete geometric
measures are%
\begin{equation*}
\begin{tabular}{lll}
$E_{\sin ^{2}}(\psi ):=1-\Lambda _{\max }^{2}({\psi })$ & and & $E_{\log
}(\psi ):=-2\log _{2}\Lambda _{\max }({\psi }),$%
\end{tabular}%
\end{equation*}%
introduced in~\cite{wg03} and~\cite{we04}, respectively. The relation of the
geometric measure to other measures has been studied in ~cite{hm06, hm06,
we08}. In the next section, we will show how the geometric measure can be
related to the permanent.

\section{Permanent and entanglement}

Here we discuss quantum states with
symmetry~\cite{hm08,hai,hu,we08,we04,we10}, in particular, permutation
symmetry. A \emph{permutation invariant basis state} is a pure state of the
form (see e.g., Refs.~\cite{we08,we04})
\begin{equation*}
\begin{tabular}{lll}
$|S(n,\vec{k})\rangle =\dfrac{\sqrt{C_{\vec{k}}^{n}}}{n!}\dsum\limits_{\pi
_{i}\in S_{n}}|\pi _{i}(\underbrace{1,..,1}_{k_{1}},\underbrace{2,..,2}%
_{k_{2}},..,\underbrace{d,..,d}_{k_{d}})\rangle ,$ & where & $C_{\vec{k}%
}^{n}:=n!/\dprod\limits_{i=1}^{d}k_{i}!.$%
\end{tabular}%
\end{equation*}

As an example, we consider the permutation invariant basis state%
\begin{equation*}
|S(4,\left( 2,2\right) )\rangle =\frac{\sqrt{6}}{4!}\sum_{\pi _{i}\in
S_{4}}|\pi _{i}(\underbrace{1,1}_{2},\underbrace{2,2}_{2})\rangle =\frac{1}{%
\sqrt{6}}\left( |1122\rangle +|1212\rangle +|2112\rangle +|1221\rangle
+|2121\rangle +|2211\rangle \right) .
\end{equation*}%
Here $n=4$ and $d=2$. Theorem \ref{th1} is our main result concerning these
class of states:

\begin{theorem}
\label{th1}Let $|S(n,\vec{k})\rangle $ be a permutation invariant basis
state. Then%
\begin{equation*}
\Lambda _{\max }\left( S(n,\vec{k})\right) =\sqrt{\frac{n!}{%
\prod_{i=1}^{d}k_{i}!}}\prod_{i=1;k_{i}\neq 0}^{d}\left( \frac{k_{i}}{n}%
\right) ^{\frac{k_{i}}{2}}.
\end{equation*}
\end{theorem}

\begin{proof}
We prove the statement by comparing the possibly entangled state $|S(n,\vec{k%
})\rangle $ to the general product states%
\begin{equation*}
\begin{tabular}{lll}
$|\phi \rangle =\dbigotimes\limits_{j=1}^{n}\left( \dsum\limits_{l=1}^{d}{%
\alpha }_{j,l}|l\rangle _{j}\right) $ & with & $\dsum\limits_{l=1}^{d}|%
\alpha _{j,l}|^{2}=1.$%
\end{tabular}%
\end{equation*}%
According to the definition of geometric measure, the first step is to
evaluate the overlap
\begin{equation*}
\phi _{\vec{k}}:=\langle S(n;\vec{k})|\phi \rangle ,
\end{equation*}%
which gives
\begin{equation}
\phi _{\vec{k}}=\dfrac{\sqrt{C_{\vec{k}}^{n}}}{n!}\sum_{\pi _{i}\in
S_{n}}\alpha _{\pi _{i}(1),1}\dots \alpha _{\pi _{i}(k_{1}),1}\alpha _{\pi
_{i}(k_{1}+1),2}\dots \alpha _{\pi _{i}(n),d}=\dfrac{\sqrt{C_{\vec{k}}^{n}}}{%
n!}\text{per}(A_{\vec{k}}),  \label{perman}
\end{equation}%
where $A_{\vec{k}}$ is an $n\times n$ matrix defined as follows:
\begin{equation*}
A_{\vec{k}}:=\left[ \underbrace{%
\begin{pmatrix}
\alpha _{1,1} \\
\alpha _{2,1} \\
\vdots \\
\alpha _{n,1}%
\end{pmatrix}%
\cdots
\begin{pmatrix}
\alpha _{1,1} \\
\alpha _{2,1} \\
\vdots \\
\alpha _{n,1}%
\end{pmatrix}%
}_{k_{1}}\underbrace{%
\begin{pmatrix}
\alpha _{1,2} \\
\alpha _{2,2} \\
\vdots \\
\alpha _{n,2}%
\end{pmatrix}%
\cdots
\begin{pmatrix}
\alpha _{1,2} \\
\alpha _{2,2} \\
\vdots \\
\alpha _{n,2}%
\end{pmatrix}%
}_{k_{2}}\cdots \underbrace{%
\begin{pmatrix}
\alpha _{1,d} \\
\alpha _{2,d} \\
\vdots \\
\alpha _{n,d}%
\end{pmatrix}%
\cdots
\begin{pmatrix}
\alpha _{1,d} \\
\alpha _{2,d} \\
\vdots \\
\alpha _{n,d}%
\end{pmatrix}%
}_{k_{d}}\right] .
\end{equation*}%
The matrix $A_{\vec{k}}$ has $k_{i}$ identical columns $\vec{v}_{i}=(\alpha
_{1,i},\alpha _{2,i},\dots ,\alpha _{n,i})^{T}$ with $i=1,\dots ,d$ and $%
\sum_{i=1}^{d}k_{i}=n$. The next step consists of maximizing the absolute
value of the overlap, \emph{i.e.}, $|\phi _{\vec{k}}|$, over the set of all $%
\alpha _{j,l}$'s such that $\sum_{l=1}^{d}|\alpha _{j,l}|^{2}=1$. On the
basis of Eq. (\ref{perman}), we can write%
\begin{equation*}
\Lambda _{\max }\left( S(n,\vec{k})\right) =\max_{\alpha _{j,l}}\frac{\sqrt{%
C_{\vec{k}}^{n}}}{n!}\left\vert \text{per}(A_{\vec{k}})\right\vert .
\end{equation*}%
To deal with this equation, we will use a recent result by Carlen, Loss and
Lieb \cite{clb06} (see also Samorodnitsky~\cite{sa08}). For any matrix $F$
defined as
\begin{equation*}
F:=[\vec{f}_{1},\vec{f}_{2},\dots ,\vec{f}_{n}],
\end{equation*}%
where $\vec{f}_{1},\vec{f}_{2},\dots ,\vec{f}_{n}$ are arbitrary column
vectors of dimension $n$, they have shown that
\begin{equation}
|\text{per}(F)|\leq \frac{n!}{n^{n/2}}\prod_{i=1}^{n}||\vec{f}_{i}||_{2},
\label{main}
\end{equation}%
where $||\vec{f}_{i}||_{2}$ denotes the $L_{2}$-norm. The r.h.s. of the
inequality can also be regarded as the permanent of a matrix whose $i$-th
column contains only identical entries $|\vec{f}_{i}|/\sqrt{n}$ . For
example,
\begin{equation*}
\frac{1}{\sqrt{n}}\left[
\begin{pmatrix}
|\vec{f}_{1}| \\
|\vec{f}_{1}| \\
\vdots \\
|\vec{f}_{1}|%
\end{pmatrix}%
\begin{pmatrix}
|\vec{f}_{2}| \\
|\vec{f}_{2}| \\
\vdots \\
|\vec{f}_{2}|%
\end{pmatrix}%
\cdots
\begin{pmatrix}
|\vec{f}_{n}| \\
|\vec{f}_{n}| \\
\vdots \\
|\vec{f}_{n}|%
\end{pmatrix}%
\right] .
\end{equation*}%
By applying the inequality in Eq. (\ref{main}), we obtain%
\begin{equation}
|\phi _{\vec{k}}|\leq \frac{\sqrt{C_{\vec{k}}^{n}}}{n!}\text{perm}(\overline{%
A}_{\vec{k}})=\sqrt{C_{\vec{k}}^{n}}\prod_{i=1}^{n}\bar{\alpha}_{i}^{k_{i}},
\label{some}
\end{equation}%
where%
\begin{equation}
\bar{\alpha}_{l}:=\sqrt{\frac{1}{n}\sum_{j=1}^{n}|\alpha _{j,l}|^{2}},
\label{avera}
\end{equation}%
with the property that $\sum_{i=1}^{d}\bar{\alpha}_{i}^{2}=1$ and
\begin{equation*}
\overline{A}_{\vec{k}}:=\left[ \underbrace{%
\begin{pmatrix}
\bar{\alpha}_{1} \\
\bar{\alpha}_{1} \\
\vdots \\
\bar{\alpha}_{1}%
\end{pmatrix}%
\cdots
\begin{pmatrix}
\bar{\alpha}_{1} \\
\bar{\alpha}_{1} \\
\vdots \\
\bar{\alpha}_{1}%
\end{pmatrix}%
}_{k_{1}}\underbrace{%
\begin{pmatrix}
\bar{\alpha}_{2} \\
\bar{\alpha}_{2} \\
\vdots \\
\bar{\alpha}_{2}%
\end{pmatrix}%
\cdots
\begin{pmatrix}
\bar{\alpha}_{2} \\
\bar{\alpha}_{2} \\
\vdots \\
\bar{\alpha}_{2}%
\end{pmatrix}%
}_{k_{2}}\cdots \underbrace{%
\begin{pmatrix}
\bar{\alpha}_{d} \\
\bar{\alpha}_{d} \\
\vdots \\
\bar{\alpha}_{d}%
\end{pmatrix}%
\cdots
\begin{pmatrix}
\bar{\alpha}_{d} \\
\bar{\alpha}_{d} \\
\vdots \\
\bar{\alpha}_{d}%
\end{pmatrix}%
}_{k_{d}}\right] .
\end{equation*}%
By Eq. (\ref{some}), we have%
\begin{equation*}
\Lambda _{\max }\left( S(n,\vec{k})\right) =\max_{|\phi \rangle }\phi _{\vec{%
k}}\leq \max_{\bar{\alpha}_{i}\in \mathbb{R}^{+}}\sqrt{C_{\vec{k}}^{n}}%
\prod_{i=1}^{n}\bar{\alpha}_{i}^{k_{i}}=\max_{\phi _{S}}\langle \phi
_{S}|S(n;\vec{k})\rangle ,
\end{equation*}%
where
\begin{equation}
|\phi _{S}\rangle =\dbigotimes\limits_{j=1}^{n}\left( \sum_{l=1}^{d}\bar{%
\alpha}_{l}|l\rangle \right) _{j}=\sum_{\vec{k}}\sqrt{C_{\vec{k}}^{n}}\,\bar{%
\alpha}_{1}^{k_{1}}\dots \bar{\alpha}_{d}^{k_{d}}|S(n;\vec{k})\rangle .
\label{pro}
\end{equation}%
The interpretation is that there is a product state constructed from $|\phi
\rangle $ by appropriately averaging the coefficients as in Eq. (\ref{avera}%
) and that the derived product state has a larger overlap. The resulting $%
|\phi _{S}\rangle $ is a tensor product of $n$ copies of the same state. The
last step is a simple maximization procedure. We need to maximize the
function%
\begin{equation*}
f\left( x_{1},x_{2},...,x_{d}\right) =\sqrt{C_{\vec{k}}^{n}}%
\prod_{i=1}^{n}x_{i}^{k_{i}}
\end{equation*}%
with nonnegative domain, under the constraint $\sum_{i=1}^{d}x_{i}^{2}=1$.
This gives
\begin{equation*}
\Lambda _{\max }\left( S(n,\vec{k})\right) =\sqrt{C_{\vec{k}}^{n}}%
\,\prod_{i=1;k_{i}\neq 0}^{d}\left( \frac{k_{i}}{n}\right) ^{\frac{k_{i}}{2}%
},
\end{equation*}%
which verifies the statement.
\end{proof}

Note that in order to find $\Lambda _{\max }(|S(n,\vec{k})\rangle )$, it is
sufficient to use the state $|\phi \rangle $ to be the product of $n$
identical copies of an arbitrary single-party state $|\alpha \rangle $,
\emph{i.e.},
\begin{equation}
|\phi \rangle =\otimes _{i=1}^{n}|\alpha \rangle .  \label{form}
\end{equation}

When the number of levels $d$ is equal to the number of parties $n$ and $%
k_{i}=1$, for every $i=1,2,...,n$, we have%
\begin{equation*}
A_{\vec{k}}:=\left[
\begin{pmatrix}
\alpha _{1,1} \\
\alpha _{2,1} \\
\vdots \\
\alpha _{d,1}%
\end{pmatrix}%
\begin{pmatrix}
\alpha _{1,2} \\
\alpha _{2,2} \\
\vdots \\
\alpha _{d,2}%
\end{pmatrix}%
\cdots
\begin{pmatrix}
\alpha _{1,d} \\
\alpha _{2,d} \\
\vdots \\
\alpha _{d,d}%
\end{pmatrix}%
\right] .
\end{equation*}%
The form of the matrix $A_{\vec{k}}$ is generic. The only constraint is that
each column has unit norm. In this case,
\begin{equation*}
\Lambda _{\max }(S(d,(\underset{d}{\underbrace{1,1,...,1}})))=\sqrt{d!}%
\left( \frac{1}{d}\right) ^{\frac{d}{2}}.
\end{equation*}

For the case of qubits, namely when $d=2$, we can prove the corresponding
result without using the Carlen-Lieb-Loss inequality, but the older Schwarz
and McClaurin inequalities. In this case, the permutation invariant states
have the form
\begin{equation*}
\begin{tabular}{lll}
$|S(n,k)\rangle =\frac{1}{\sqrt{C_{k}^{n}}}\sum_{\pi _{i}\in S_{n}}|\pi _{i}(%
\underbrace{0,..,0}_{k},\underbrace{1,..1)}_{n-k}\rangle ,$ & where & $%
C_{k}^{n}:=\frac{n!}{k!\left( n-k\right) !}.$%
\end{tabular}%
\end{equation*}%
The theorem below states the connected result:

\begin{theorem}
\label{th2}Let $|S(n,k)\rangle $ be a permutation invariant basis state for
qubits. Then%
\begin{equation*}
\Lambda _{\max }(n,k)=\sqrt{\frac{n!}{k!\left( n-k\right) !}}\left( \frac{k}{%
n}\right) ^{\frac{k}{2}}\left( \frac{n-k}{n}\right) ^{\frac{n\!-\!k}{2}}
\end{equation*}
\end{theorem}

\begin{proof}
We want to find the maximal overlap between $|S(n,k)\rangle $ with product
states
\begin{equation*}
|\phi \rangle =\otimes _{j=1}^{n}(\sqrt{q_{j}}|0\rangle +\sqrt{1-q_{j}}%
e^{i\beta _{j}}|1\rangle ).
\end{equation*}%
As the coefficients in $|S(n,k)\rangle $ are nonnegative, we can set $\beta
_{j}=0$. We then evaluate%
\begin{equation*}
\phi _{k}\equiv \langle S(n,k)|\phi \rangle =\frac{\sqrt{C_{k}^{n}}}{n!}%
\sum_{\pi _{i}\in S_{n}}\left( \prod_{l=1}^{k}\sqrt{q_{\pi _{i}(l)}}%
\prod_{l=k+1}^{n}\sqrt{1-q_{\pi _{i}(k+1)}}\right) .
\end{equation*}%
Using the Cauchy-Schwarz inequality, we have
\begin{equation*}
|\phi _{k}|^{2}\leq \frac{C_{k}^{n}}{(n!)^{2}}\left( \sum_{\pi _{i}\in
S_{n}}\prod_{l=1}^{k}q_{\pi _{i}(l)}\right) \left( \sum_{\pi _{i}\in
S_{n}}\prod_{l=k+1}^{n}1-q_{\pi _{i}(k+1)}\right) .
\end{equation*}%
By the Maclaurin inequality%
\begin{equation*}
\frac{1}{n!}\sum_{\pi _{i}\in S_{n}}\prod_{l=1}^{k}x_{\pi _{i}(l)}\leq
\left( \frac{1}{n}\sum_{i=1}^{n}x_{i}\right) ^{k},
\end{equation*}%
we arrive at
\begin{equation*}
|\phi _{k}|^{2}\leq C_{k}^{n}(\bar{q})^{k}(1-\bar{q})^{n-k}=C_{k}^{n}\cos
^{2k}\theta \sin ^{2(n-k)}\theta ,
\end{equation*}%
for $\cos ^{2}\theta =\bar{q}$. This means that
\begin{equation*}
|\phi _{k}|\leq \sqrt{C_{k}^{n}}(\bar{q})^{k/2}(1-\bar{q})^{(n-k)/2}=\sqrt{%
C_{k}^{n}}\cos ^{k}\theta \sin ^{(n-k)}\theta .
\end{equation*}%
Maximizing the expression on the r.h.s. over the angle $\theta $, we obtain
\begin{eqnarray*}
\Lambda _{\max }(n,k) &=&\max_{\theta }\sqrt{C_{k}^{n}}\cos ^{k}\theta \sin
^{(n-k)}\theta \\
&=&\sqrt{\frac{n!}{k!\left( n-k\right) !}}\left( \frac{k}{n}\right) ^{\frac{k%
}{2}}\left( \frac{n-k}{n}\right) ^{\frac{n\!-\!k}{2}}.
\end{eqnarray*}%
This concludes the proof.
\end{proof}

For the case of three qubits, some examples of permutation invariant states
are the following ones:
\begin{equation*}
\begin{tabular}{lll}
$|\mathrm{GHZ}\rangle \equiv (|000\rangle +|111\rangle )/\sqrt{2},$ & $|%
\mathrm{W}\rangle \equiv (|001\rangle +|010\rangle +|100\rangle )/\sqrt{3},$
& $|\mathrm{\bar{W}}\rangle \equiv (|110\rangle +|101\rangle +|011\rangle )/%
\sqrt{3}.$%
\end{tabular}%
\end{equation*}%
The states $|\mathrm{GHZ}\rangle $ and $|\mathrm{W}\rangle $ have extremal
properties and have particularly important roles in quantum mechanics \cite%
{horo}: the \emph{Greenberger--Horne--Zeilinger state}, $|\mathrm{GHZ}%
\rangle $, was used to test Bell's inequalities; the \emph{W state}, $|%
\mathrm{W}\rangle $, exhibits genuine three-party entanglement, in a
different way from $|\mathrm{GHZ}\rangle $. By applying Theorem \ref{th2},
we have
\begin{equation*}
\begin{tabular}{lll}
$\Lambda _{\max }(\mathrm{GHZ})=1/\sqrt{2}$ & and & $\Lambda _{\max }(%
\mathrm{W})=\Lambda _{\max }(\mathrm{\bar{W}})=2/3.$%
\end{tabular}%
\end{equation*}

Let us now consider examples for three-party $4$-level systems, namely $n=3$
and $d=4$. The chosen vectors are $\vec{a}=(2,0,0,1)$ and $\vec{b}=(1,1,1,0)$
and their corresponding states are
\begin{equation*}
\begin{tabular}{lll}
$|\vec{a}\rangle \equiv \frac{1}{\sqrt{3}}(|114\rangle +|141\rangle
+|411\rangle ),$ & and & $|\vec{b}\rangle \equiv \frac{1}{\sqrt{6}}%
(|123\rangle +|132\rangle +|213\rangle +|231\rangle +|312\rangle
+|321\rangle ).$%
\end{tabular}%
\end{equation*}%
With the use of Theorem \ref{th1}, we have
\begin{equation*}
\begin{tabular}{lll}
$\Lambda _{\max }(\vec{a})=\Lambda _{\max }(\mathrm{W})=2/3$ & and & $%
\Lambda _{\max }(\vec{b})=\sqrt{2}/3$.%
\end{tabular}%
\end{equation*}%
Note that the states $|\vec{a}\rangle $ and $|\mathrm{W}\rangle $ have the
same structure.

We conclude by asking the following question and then providing a partial
answer:

\begin{conjecture}
\label{problem}Is it true that in order to obtain the maximal overlap of any
permutation invariant state, we can assume the product state to be a tensor
product of the same single-party state?
\end{conjecture}

As we have pointed out in the historical note in the Introduction, this answer
is yes. Here we give simple proof for certain classes of permutation invariant
states, beyond the basis states discussed above. For the general proof, we
refer to Refs.~\cite{di,hu}.

The first class we consider is the case where the coefficients $c_{\vec{k}}$%
's are nonnegative. To obtain the maximal overlap for the corresponding
state $|\psi \rangle $, we can as well set the coefficients in the
unentangled state $|\phi \rangle $ to be nonnegative, as the goal is to
maximize the overlap between $|\psi \rangle $ and $|\phi \rangle $. Thus we
have
\begin{equation*}
\langle \phi |\psi \rangle =\sum_{\vec{k}}c_{\vec{k}}\phi _{\vec{k}}=\sum_{%
\vec{k}}c_{\vec{k}}\frac{\sqrt{C_{\vec{k}}^{n}}}{n!}\text{per}(A_{\vec{k}%
})\leq \sum_{\vec{k}}c_{\vec{k}}\frac{\sqrt{C_{\vec{k}}^{n}}}{n!}\text{per}(%
\overline{A}_{\vec{k}})=\sum_{\vec{k}}c_{\vec{k}}\sqrt{C_{\vec{k}}^{n}}\bar{%
\alpha}_{1}^{k_{1}}\dots \bar{\alpha}_{d}^{k_{d}}=\langle \phi _{S}|\psi
\rangle ,
\end{equation*}%
where we have used again the inequality proved in \cite{clb06}. This means
that for nonnegative $c_{\vec{k}}$'s, in order to maximize the overlap, we
can use the product state $|\phi _{S}\rangle $ in Eq. (\ref{pro}),
consisting of a direct product of identical single-party states. An example
of this is given by the states%
\begin{equation*}
\begin{tabular}{lll}
$|\mathrm{W\bar{W}}\left( s\right) \rangle \equiv \sqrt{s}|\mathrm{W}\rangle
+\sqrt{1-s}|\mathrm{\bar{W}}\rangle ,$ & where & $s\in \lbrack 0,1].$%
\end{tabular}%
\end{equation*}%
The way to use the product state in the form of Eq. (\ref{form}) is
justified. This was the \emph{ansatz} used to calculate the entanglement for
this family of states in \cite{wg03}. In particular, the product state can
be written as
\begin{equation*}
|\phi _{S}\left( \theta \right) \rangle \equiv \left( \cos \theta |0\rangle
+\sin \theta |1\rangle \right) ^{\otimes 3}
\end{equation*}%
and then, by maximizing the inner product $\langle \phi _{S}\left( \theta
\right) |\mathrm{W\bar{W}}\left( s\right) \rangle $ over $\theta $, it is
straighforward to obtain the expression%
\begin{equation*}
\Lambda _{\max }\left( \mathrm{W\bar{W}}\left( s\right) \right) =\frac{1}{2}%
\left( \sqrt{s}\cos \theta (s)+\sqrt{1-s}\sin \theta \left( s\right) \right)
\sin 2\theta (s),
\end{equation*}%
where $\theta \left( s\right) $ is the solution of the equation%
\begin{equation*}
\begin{tabular}{lll}
$\sqrt{1-s}\tan ^{3}\theta +2\sqrt{s}\tan ^{2}\theta -2\sqrt{1-s}\tan \theta
-\sqrt{s}=0,$ & where & $\tan \theta \in \left[ 1/\sqrt{2},\sqrt{2}\right] $%
\end{tabular}%
.
\end{equation*}

We can also approach a more general class of states. When the coefficient $%
c_{\vec{k}}$' s are arbitrary but nonnegative, the above consideration of
using states as in Eq. (\ref{form}) holds; this gives us the corresponding
state
\begin{equation*}
|\psi \rangle =\sum_{\vec{k}}c_{\vec{k}}|S(n,\vec{k})\rangle .
\end{equation*}%
Now, we perform a basis change on $|\psi \rangle $, \emph{i.e.},
\begin{equation*}
|\psi ^{\prime }\rangle \equiv (U\otimes U\otimes ...\otimes U)|\psi \rangle
=\sum_{\vec{k}}b_{\vec{k}}|S(n,\vec{k})\rangle ,
\end{equation*}%
where $U$ is any unitary transformation in U$(d)$. Specifically, the
transformation $U$ acts on a single party, and $b_{\vec{k}}$'s are the
resulting coefficients for $|\psi ^{\prime }\rangle $ expanded in the basis
of $|S(n,\vec{k})\rangle $. The resulting coefficients $b_{\vec{k}}$'s are
in general complex. Since we have shown that to calculate the entanglement
for $|\psi \rangle $ we can assume the product state to be as in Eq. (\ref%
{form}) and $|\psi ^{\prime }\rangle $ is simply given by a local change of
basis, to calculate the entanglement for $|\psi ^{\prime }\rangle $, we can
take a fiducial state of the same form. To illustrate this fact, we consider
a generic element of SU$\left( 2\right) $:%
\begin{equation*}
\begin{tabular}{lll}
$U=\left(
\begin{array}{cc}
u & v \\
-v^{\ast } & u^{\ast }%
\end{array}%
\right) ,$ & where & $|u|^{2}+|v|^{2}=1.$%
\end{tabular}%
\end{equation*}%
The effect of $U$ on the parties of $|\mathrm{W}\rangle $ is given by
\begin{equation*}
|\mathrm{W}\rangle \longmapsto \frac{1}{\sqrt{3}}\left( -3u^{2}v^{\ast
}|000\rangle +\sqrt{3}u\left( |u|^{2}-2|v|^{2}\right) |\mathrm{W}\rangle +%
\sqrt{3}v\left( 2|u|^{2}-|v|^{2}\right) |\mathrm{\bar{W}}\rangle +3u^{\ast
}v^{2}|111\rangle \right)
\end{equation*}%
and similarly for $|\mathrm{\bar{W}}\rangle $, $|000\rangle $ and $%
|111\rangle $, our basis states in the symmetric subspace. The corresponding
coefficients are in general complex. When $U$ is diagonal, the coeffients $%
c_{\vec{k}}$ are transformed as $d_{\vec{k}}=c_{\vec{k}}e^{i\vec{k}\cdot
\vec{\theta}}$, where $\vec{\theta}$ is an arbitrary real $d$-component
vector characterizing the matrix $U$.

It takes $8$ real parameters to describe the generic permutation invariant
states for three qubits. If we start with $4$ arbitrary nonnegative
coefficients and supplement with arbitrary U$(2)$ transformations (which have
$4$ real parameters), we have then $8$ real parameters in total. The number
boils down to $6$ in both cases, if we take into account normalization and
global phase. This counting suggests that the statement in
Conjecture~\ref{problem} may be true, which we now know it is correct. We note
that there have been applications using this to search for maximal entangled
states; see, e.g., Ref.~\cite{ta}, in which the above GHZ and W states belong
to such a family.

\begin{problem}
\label{problem2} Given that the answer to Conjecture 1 is correct, can we then
derive a more general permanent inequality than that of Carlen, Lieb and Loss?
\end{problem}

We give our reasoning of why such a general inequality might exist. The
inequality of Carlen, Lieb and Loss has enabled us to prove a partial answer
to Conjecture 1 for nonnegative symmetric states, but not the most generality.
The more general permanent inequality if it existed would enable us to
establish the proof of Conjecture 1 in the most generality and would be such
an inequality that we seek. Since Conjecture 1 is correct, we expect a more
general permanent inequality might exist.

Given the basis $|{S(n,\vec{k})}\rangle$, the most general permutation
invariant state $|{\psi}\rangle$ can be expanded as
\begin{equation}
|{\psi}\rangle=\sum_{\vec{k}} q_{\vec{k}}|{S(n,\vec{k})}\rangle.
\end{equation}
For the general product state $|{\phi}\rangle$ in Theorem 1, we have seen that
their overlap can be expressed in terms of linear combination of permanents,
\begin{equation}
\langle{\psi}|{\phi}\rangle=\sum_{\vec{k}} q_{\vec{k}}^*
\frac{\sqrt{C^n_{\vec{k}}}}{n\!} {\rm per}(A_{\vec{k}}).
\end{equation}
Since the answer to Conjecture 1 is affirmative, this implies that we can
always find a product state $|{\bar{\phi}}\rangle=\otimes_{j=1}^n
(\bar{\alpha}_l|{l}\rangle_j)$ such that
\begin{equation}
|\langle{\psi}|{\phi}\rangle|\le |\langle{\psi}|{\bar{\phi}}\rangle|,
\end{equation}
which is equivalent to
\begin{equation}
\Big|\sum_{\vec{k}} q_{\vec{k}}^* \frac{\sqrt{C^n_{\vec{k}}}}{n\!} {\rm
per}(A_{\vec{k}})\Big|\le \Big|\sum_{\vec{k}} q_{\vec{k}}^*
\frac{\sqrt{C^n_{\vec{k}}}}{n\!} {\rm per}(\bar{A}_{\vec{k}})\Big|,
\end{equation}
where $q_{\vec{k}}$'s are the coefficients and are arbitrary (up to
normalization constraint). If we could express $\bar{\alpha}_l$'s in terms of
$\alpha_{j,l}$'s which characterize $|{\phi}\rangle$, we would then achieve
our goal for a more general inequality of permanents. We leave this as an open
question.

\end{document}